\documentclass[showpacs,preprintnumbers,amsmath,amssymb]{revtex4}

\usepackage{graphicx}
\usepackage{dcolumn}
\usepackage{bm}

\usepackage{xcolor}

\usepackage{amsmath,amssymb}

\usepackage{amsthm}
\usepackage{enumerate}

\DeclareMathOperator{\e}{\mathrm{e}}
\DeclareMathOperator{\dd}{\mathrm{d\!}}
\newcommand{\df}[2]{\frac{\dd #1}{\dd #2}}

\newtheorem{tw}{Theorem}

\newtheorem{cor}[tw]{Corollary}
\newtheorem{prop}[tw]{Proposition}

\newtheorem{rem}[tw]{Remark}

\begin{document}
\preprint{APS/123-QED}

\title{Evolution of populations with strategy-dependent time delays}

\author{Jacek Mi\c{e}kisz}
\email{miekisz@mimuw.edu.pl}
\author{Marek Bodnar}%
\email{mbodnar@mimuw.edu.pl}
\affiliation{%
    Institute of Applied Mathematics and Mechanics,\\
   University of Warsaw, Warsaw, Poland
}%

\date{\today}

\begin{abstract}
We study effects of strategy-dependent time delays on equilibria of evolving populations. It is well known that time delays may cause oscillations in dynamical systems.  
Here we report a novel behavior. We show that microscopic models of evolutionary games with strategy-dependent time delays lead to a new type of replicator dynamics. 
It describes the time evolution of fractions of the population playing given strategies and the size of the population.
Unlike in all previous models, stationary states of such dynamics depend continuously on time delays. 
We show that in games with an interior stationary state (a globally asymptotically stable equilibrium in the standard replicator dynamics), at certain time delays, 
it may disappear or there may appear another interior stationary state. In the Prisoner's Dilemma game, for time delays of cooperation 
smaller than time delays of defection, there appears an ustable interior equilibrium and therefore for some initial conditions, 
the population converges to the homogeneous state with just cooperators.
\end{abstract}

\pacs{Valid PACS appear here}
\maketitle   
   
\section{Introduction}

Many social and biological processes can be modeled as systems of interacting individuals within the framework of evolutionary game theory 
\cite{maynard1,maynard2,weibull,redondo,hofbauerbook,young,cressman,nowak,sandholm}. 
The evolution of very large (infinite) populations can be then given by differential replicator equations which describe time changes of fractions of populations playing 
different strategies \cite{taylor,hofbauer,weibull,hofbauerbook}. It is usually assumed (as in the replicator dynamics) that interactions between individuals take place instantaneously 
and their effects are immediate. In reality, all social and biological processes take a certain amount of time. Results of biological interactions between individuals 
may appear in the future, and in social models, individuals or players may act, that is choose appropriate strategies, on the basis of the information concerning events in the past.
It is natural therefore to introduce time delays into evolutionary game models. 

It is well known that time delays may cause oscillations in dynamical systems \cite{ladas,gopalsamy,kuang,erneux}.
One usually expects that interior equilibria of evolving populations, describing coexisting strategies or behaviors, are asymptotically stable
for small time delays and above a critical time delay, where the Hopf bifurcation appears, they become unstable, 
evolutionary dynamics exhibits oscillations and cycles. Here we report a novel behavior - continuous dependence of equilibria on time delays.

Effects of time delays in replicator dynamics were discussed in \cite{tao,alboszta,aoku,ijima1,ijima2,moreira,wesol,matusz,wessonrand1,wessonrand2,nesrine1,nesrine2} 
for games with an interior stable equilibrium (an evolutionarily stable strategy \cite{maynard1,maynard2}). 
In \cite{tao}, the authors discussed the model, where individuals at time $t$ imitate a~strategy with a~higher average payoff 
at time $t-\tau$ for some time delay $\tau$. They showed that the interior stationary state of the resulting time-delayed differential equation 
is locally asymptotically stable for small time delays and for big ones it becomes unstable, there appear oscillations. 
In \cite{alboszta}, we constructed a different type of a model, where individuals are born $\tau$ units of time after their parents played. 
Such a model leads to system of equations for the frequency of the first strategy and the size of the population. 
We showed the absence of oscillations -- the original stationary point is globally asymptotically stable for any time delay. 
In both models, the position of the equilibrium is not affected by time delays.

Here we modify the second of the above models by allowing time delays to depend on strategies played by individuals,
we observe a new behavior of a dependence of equilibria on delays.   

Recently there were studied models with strategy-dependent time delays. In particular, Moreira et al. \cite{moreira} discussed a multi-player Stag-hunt game, 
Ben Khalifa et al. \cite{nesrine1} investigated asymmetric games in interacting communities, Wesson and Rand \cite{wessonrand1} studied Hopf bifurcations in two-strategy delayed replicator dynamics.
The authors generalized the model presented in \cite{tao}, they studied asymptotic stability of equilibria and the presence of bifurcations. Some very specific examples 
of three-player games with strategy-dependent time delays were studied in \cite{bmv}, a shift of interior equilibria was observed.

Here we present a systematic study of effects of strategy-dependent time delays on the long-run behavior of two-player games with two strategies in models which are generalizations of the one in \cite{alboszta}. We consider Stag-hunt-type games with two basins of attraction of pure strategies, Snowdrift-type games with a stable interior equilibrium (a coexistence of two strategies), 
and the Prisoner's Dilemma game. We report a novel behavior. We show that stationary states depend continuously on time delays. Moreover, at certain time delays, 
an interior stationary state may disappear or there may appear another interior stationary state. 

Below we present a general theory and particular examples. In the Appendix we provide proofs and additional theorems, 
in particular general conditions for the existence and uniqueness of interior states in Snowdrift-type games (Theorems A.3 and~A.6).

\section{Methods}
\subsection*{A. Replicator dynamics}
We assume that our populations are haploid, that is the offspring have identical phenotypic strategies as their parents.
We consider symmetric two-player games with two strategies, $C$~and $D$, given by the following payoff matrix:

\vspace{3mm}

\hspace{23mm} $C$  \hspace{2mm} $D$   

\hspace{15mm} $C$ \hspace{3mm} $a$  \hspace{3mm} $b$ 

$U$ = \hspace{6mm} 

\hspace{15mm} $D$ \hspace{3mm} $c$  \hspace{3mm} $d$,

\noindent where the $ij$ entry, $i,j = C, D$, is the payoff of the first (row) player when
it plays the strategy $i$ and the second (column) player plays the strategy $j$. 
We assume that both players are the same and hence payoffs of the column player are given 
by the matrix transposed to $U$; such games are called symmetric. 

Below we will consider all three main types of two-player games with two strategies: Stag-hunt, Snowdrift, and Prisoner's Dilemma game. 
They serve as simple models of social dilemmas. Strategies, C and D, may be interpreted as cooperation and defection. 

Let us assume that during a~time interval of the length $\varepsilon$, only an $\varepsilon$-fraction 
of the population takes part in pairwise competitions, that is plays games.
Let $p_{i}(t)$, $i=C, D,$ be the number of individuals playing at the time $t$ the strategy 
$C$ and $D$ respectively, $p(t)=p_{C}(t)+p_{D}(t)$ the total number of players, and
$x(t)=\frac{p_{C}(t)}{p(t)}$ the fraction of the population playing $C$.
Let 
\begin{equation}\label{wyplaty}
U_{C}(t)= ax(t)+b(1-x(t)) \quad \text{ and } \quad U_{D}(t)= cx(t)+d(1-x(t))  
\end{equation}
be average payoffs of individuals playing $C$~and $D$ respectively.

Now we would like to take into account that individuals are born some units of time after their parents played. 
We assume that time delays depend on strategies and are equal to $t-\tau_{C}$ or $t-\tau_{D}$ respectively. 

We propose the following equations:
\begin{align}
   p_{i}(t + \varepsilon) &= (1-\varepsilon)p_{i}(t) + \varepsilon p_{i}(t-\tau_{i})U_{i}(t-\tau_{i}); \qquad i = C, D, \label{bdiscrete1}\\
   p(t + \varepsilon)  &= (1-\varepsilon)p(t) + \varepsilon \Bigl(p_{C}(t-\tau_{C})U_{C}(t-\tau_{C}) + p_{D}(t-\tau_{D})U_{D}(t-\tau_{D})\Bigr).\label{bdiscrete2}
\end{align}

In the above we assume that populations are large in order to justify a continuous description of population sizes 
but they are not infinite as in classical replicator equations. The parameter $\varepsilon$ represents the time length of game interactions and therefore it multiplies the second terms in the above equations.
We assume the replacement of parents or equivalently a death rate $1$ and hence $p_{i}(t)$ and $p(t)$ are multiplied by $1-\varepsilon$. 

We divide (\ref{bdiscrete1}) by (\ref{bdiscrete2}) for $i=A$, obtain the equation for  $x(t+\varepsilon)\equiv x_{C}(t+\varepsilon)$, subtract $x(t)$, divide the difference by $\varepsilon$, 
take the limit $\varepsilon \rightarrow 0$, and get an equation for the frequency of the first strategy,
\begin{equation}\label{bxeq}
\frac{\dd x}{\dd t}= \frac{p_{C}(t-\tau_{C})U_{C}(t-\tau_{C})(1-x(t)) - p_{D}(t-\tau_{D})U_{D}(t-\tau_{D})x(t)}{p(t)}
\end{equation}
which can also be written as 
\begin{equation}\label{bxeqb}
\frac{\dd x}{\dd t}= \frac{x(t-\tau_{C})p(t-\tau_{C})U_{C}(t-\tau_{C}))(1-x(t)) - (1-x(t-\tau_{D}))p(t-\tau_{D})U_{D}(t-\tau_{D})x(t)}{p(t)}.
\end{equation}

Let us notice that unlike in the standard replicator dynamics, the above equation for the frequency of the first strategy is not closed, there appears in it a variable describing 
the size of the population at various times. One needs equations for populations sizes. From \eqref{bdiscrete1} and \eqref{bdiscrete2} we get

\begin{align}
   \frac{\dd p_{i}(t)}{\dd t} &= -p_{i}(t) +  p_{i}(t-\tau_{i})U_{i}(t-\tau_{i}); \; \; i~= C, D, \label{bcont1}\\
   \frac{\dd p(t)}{\dd t} &= -p(t) + \Bigl(p_{C}(t-\tau_{C})U_{C}(t-\tau_{C}) + p_{D}(t-\tau_{D})U_{D}(t-\tau_{D})\Bigr). \label{bcont2}
\end{align}

To trace the evolution of the population, we have to solve the system of equations, ((\ref{bxeqb}),(\ref{bcont2})),
together with initial conditions on the interval $[-\tau_M,0]$, where 
$\tau_M = \max\bigl\{\tau_{C},\tau_{D}\bigr\}$. We assume that 
\begin{equation}\label{bdic}
   x(t) = \varphi_x(t), \quad p(t) = \varphi_p(t), \quad \text{ for } \quad t\in [-\tau_M,0].
\end{equation}

We have the following proposition concerning the existence of non-negative solutions.

\begin{prop}
   If the initial functions $\varphi_x$ and $\varphi_p$ are continuous on $[-\tau_M,0)$ and non-negative, then there exists a unique non-negative 
   solution of the system~(\eqref{bxeqb}, \eqref{bcont2}) with initial conditions~\eqref{bdic} well defined on the interval $[0,+\infty)$.
\end{prop}
\begin{proof}
   The local existence of the solution follows immediately from a standard theory of delay differential equations, ~\cite{kuang}. The non-negativity 
   follows from~\cite{bodnar1}. To prove the global existence it is enough to use the step method and to observe that on the interval
   $\bigl[0,\tau_M\bigr]$ the system~(\eqref{bxeqb}, \eqref{bcont2}) becomes a system of non-autonomous ordinary differential equations. The 
   equation for $p$ becomes linear, it can be solved and then the equation for $x$ becomes linear with respect to $x(t)$. 
\end{proof}

\subsection*{B. Stationary states}

We derive here an equation for an interior stationary state (the stationary frequency of the first strategy). Let us assume that there exists a stationary 
frequency $\bar x$ such that $x(t) = \bar x$ for all $t\ge 0$ and for some suitably chosen function $p(t)$. 
Then average payoffs of each strategy are constant and are equal to 
\begin{equation}\label{stat:payoff}
   \bar U_C = a \bar x + b(1-\bar x) \quad \text{ and } \quad \bar U_D = c \bar x + d (1-\bar x).
\end{equation}
Thus, equation~\eqref{bcont2} becomes a linear delay differential equation, 
\begin{equation}\label{eqss:p}
\df{p}{t} = \bar x \Bigl(a\bar x+b(1-\bar x)\Bigr) p(t-\tau_C) + (1-\bar x) \Bigl(c\bar x + d(1-\bar x)\Bigr) p(t-\tau_D) - p(t).
\end{equation}
Note, that solutions of~\eqref{eqss:p} with non-negative initial conditions are non-negative. This implies that the leading eigenvalue of 
this equation is real. The eigenvalues $\lambda$ of \eqref{eqss:p} satisfy
\begin{equation}\label{eqss:nalambda}
   \lambda +1 = \bar x \bar U_C e^{-\lambda\tau_C} + (1-\bar x) \bar U_D \e^{-\lambda\tau_D}.
\end{equation}
Assume now that $\lambda$ is a solution of~\eqref{eqss:nalambda} (of course $\lambda$ depends on $\bar x$) and $p(t) = p_0 \exp(\lambda t)$ for some $p_{0}$. 

We plug such $p$ into \eqref{bxeq} and we get two stationary solutions of~\eqref{bxeq} (i.e. such that the right-hand side of~\eqref{bxeq} is equal to 0), 
$\bar x =0, 1$, and possibly interior ones - solutions to
\begin{equation}\label{stateq}
   \bar U_C \e^{-\lambda\tau_C} = \bar U_D\e^{-\lambda\tau_D}.
\end{equation}

If $\tau_C=\tau_D$ and $d<b<a<c$, then \eqref{stateq} gives us a mixed Nash equilibrium (an evolutionarily stable strategy, \cite{maynard1,maynard2}) of the game,
\[
   x^* = \frac{b-d}{b-d+c-a}
\]
which represents the equilibrium fraction of an infinite population playing $C$ \cite{weibull,hofbauerbook}.
In the non-delayed replicator dynamics, $x^{*}$ is globally asymptotically stable and there are two unstable stationary states: $x=0$ and $x=1$.   

If $\tau_C\neq\tau_D$, then $\lambda = \ln(\bar U_C/ \bar U_D) /(\tau_C-\tau_D)$. We plug it into~\eqref{eqss:nalambda} and we conclude that 
$\bar x$ satisfies an equation $F(\bar x)=0$, where 
\begin{equation}\label{funss}
   F(x) = \frac{1}{\tau_C-\tau_D}\ln\left(\frac{a x + b(1-x)}{c x + d(1-x)}\right) +1 - 
\frac{\Bigl(c x + d(1- x)\Bigr)^{\tau_C/(\tau_C-\tau_D)}}{\Bigl(a x + b(1-x)\Bigr)^{\tau_D/(\tau_C-\tau_D)}}.
\end{equation}
From the above it follows the following proposition that explains in what sense $\bar x$ is a stationary state of the replicator dynamics (\eqref{bxeqb}, \eqref{bcont2}).
\begin{prop}
   Assume that $\tau_C\neq \tau_D$. Let $\bar x$ be a solution to $F(x)=0$, where $F$ is defined by~\eqref{funss} and let 
   \[
      \bar\lambda =  \frac{1}{\tau_C-\tau_D}\ln\left(\frac{a \bar x + b(1-\bar x)}{c \bar x + d(1-\bar x)}\right).
   \]
   Then the functions 
   \[
      x(t) = \bar x, \quad p(t) = p_0 \e^{\bar\lambda t}, \quad t\ge 0
   \]
   are solutions of the system~(\eqref{bxeqb}, \eqref{bcont2}) with the initial conditions
   \[
      \varphi_x(t) = \bar x, \quad \varphi_p(t) = p_0\e^{\bar \lambda t}, \quad \text{ for } \quad t\in [-\tau_M,0].
   \]
\end{prop}

In this paper, we consider only examples of games with a positive $\lambda$, that is our populations grow exponentially with time hence the use of differential equations
to describe time evolution is appropriate. For a negative $\lambda$, the population gets extinct.  

Further properties of the function $F$ related to the existence of zeros of this function in the interval $(0,1)$ are analyzed in Subsection~\ref{app1} in the Appendix.
In particular, we provide in Theorems A.3 and A.6 general conditions for the existence and uniqueness of the interior stationary state in Snowdrift-type games.
\vspace{2mm}

\section{Results}

\subsection{Stag-hunt games}
Here we consider games with a unique unstable interior equilibrium. We begin with a general Stag-hunt game given by the following payoff matrix:
\[
U_1 = \begin{bmatrix}
a & 0 \\ \beta a  & \beta a
\end{bmatrix},
\]
where $a>0$ and $\beta\in (0,1)$. In this case the function $F$ takes the form

\begin{equation}\label{fsh}
F(x) = 1 +\frac{\alpha}{\tau_C} \ln \frac{x}{\beta} - a \beta^{\alpha} x^{1-\alpha}, \quad \alpha = \frac{\tau_C}{\tau_C-\tau_D}.
\end{equation}

Note that if $\bar x \in(0,1)$ is an interior stationary frequency, then the leading eigenvalue is given by the following formula (see Proposition 2), 
\begin{equation}\label{lambdash}
\bar \lambda = \frac{1}{\tau_C-\tau_D} \ln \frac{\bar x}{\beta}. 
\end{equation}
Thus $\bar \lambda>0$ if and only if 
\begin{itemize}
   \item $\tau_C>\tau_D$ and $\bar x \in (\beta, 1)$ 
   \item[] or
   \item $\tau_C<\tau_D$ and $\bar x \in (0,\beta)$.
\end{itemize}
For $\tau_C>\tau_D$ it is easy to see that $F$ is increasing from $-\infty$ to $F(1)=1-a\beta^{\alpha}$. Thus, the interior stationary frequency $\bar x\in (0,1)$ exists 
if and only if $F(1)>0$ which is equivalent to the inequality $\displaystyle \beta< a^{-1/\alpha}$. The leading eigenvalue is positive if $\bar x>\beta$, 
which is equivalent to the inequality $F(\beta)<0$. This gives a condition $a\beta >1$. Finally, we get here a necessary condition for the existence of an interior stationary frequency 
with a positive leading eigenvalue:

\begin{equation}\label{exist}
\frac{1}{a}<\beta<\frac{1}{a^{1/\alpha}} \; \Longrightarrow \; a>1. 
\end{equation}
For $\tau_C<\tau_D$ it is easy to see that $F$ is decreasing from $\infty$ to $F(1)=1-a\beta^{-|\alpha|}$. Thus, the interior stationary frequency $\bar x\in (0,1)$ exists if and only if $F(1)<0$ 
which is equivalent to the inequality $\displaystyle \beta< a^{1/|\alpha|}$. Similarly, the leading eigenvalue is positive if $\bar x<\beta$, which is equivalent to the inequality $F(\beta)<0$. 
This gives a condition $a\beta >1$. Again we see that a necessary condition for the existence of an interior stationary frequency with a positive leading eigenvalue is $a>1$. 

For general Stag-hunt games we can derive an explicit formula for the interior stationary state $\bar x$ using the Lambert $W$ function $W_p$. Namely we solve $F(x)=0$, where $F$ is given by~\eqref{fsh}, 
and after some algebraic caluculations we get

\begin{equation}\label{lambertsh}
\bar x = \left(\frac{\tau_Da \beta^{\frac{\tau_C}{\tau_C-\tau_D}}}{W_p\left(\tau_D\beta a e^{\tau_D}\right)} \right)^{\frac{\tau_C-\tau_D}{\tau_D}}
\end{equation}

For $a=5$ and $\beta=3/5$ we get the classical Stag-hunt game.

In Fig.~\ref{fig.stag_hunt} we present the dependence of the unstable interior stationary state on time delays. We see that the bigger a time delay of a given strategy is, the smaller its basin of atraction. 
Moreover, if the time delay of the strategy $C$ increases, at certain point, the interior stationary state ceases to exist and therefore the strategy $D$ becomes globally asymptotically stable
(the internal equilibrium hits $x=1$, destabilizing this solution). 
For $\tau_C = \tau, \tau_D = 2\tau$ we get $\alpha = -1$ and therefore (14) simplifies and in the limit $ \tau \to \infty$ one obtains  $\bar x = \sqrt{\frac{\beta}{a}}$.
For the classical Stag-hunt game, $\bar x = \frac{\sqrt{3}}{5} \approx 0.346$ as it is seen in Panel A.
 
Panel~D of Fig.~\ref{fig.stag_hunt} indicates the existence of a curve $\tau_D^*(\tau_C^*)$ such that for $\tau_C>\tau_{C}^*$  and fixed $\tau_D=\tau_{D}^*$, there is no interior stationary frequency and for $\tau_C<\tau_{C}^*$ there is a unique one. The curve $\tau_D^*(\tau_C^*)$ that splits the plane $(\tau_C,\tau_D)$ into two regions (below this curve there is no interior stationary frequency, above it there is a unique one) is given by the following formula,
\begin{equation}\label{StagHuntGranica}
\tau_{D}^* = \tau_{C}^*\left(1+\frac{\ln \beta}{W_p(a\tau_{C}^*\e^{\tau_C^*})-\tau_C^*}\right), 
\quad \tau_C^* \ge \frac{-\ln \beta}{a\beta -1}.
\end{equation}
The formula can be easily obtained from $F(1)=0$ (detailed calculations can be found in Subsection~\ref{app2} in the Appendix). For the classical Stag-hunt game, the curve given by (18) is almost a straight line. Its derivative with respect to $\tau_C^*$ changes from 
\[
\frac{a\beta -1}{a\beta-1-a\beta \ln \beta}\approx  0.566 
\]
for $\tau_C^* = \frac{-\ln \beta}{a\beta-1} \approx 0.255$
to 
\[
1+\frac{\ln\beta}{\ln a} \approx 0.683
\]
for $\tau_C^*\to \infty$. 

Note also, that if $\tau_C$ is fixed, and we take $\tau_D\to+\infty$, then we get $\alpha\to 0$ and the formula~\eqref{fsh} simplifies to 
$ 1-a x$, which implies that the interior stationary frequency for sufficiently large $\tau_D$ is close to $\frac{1}{a}$. 

Let us summarize the results. When we parametrize both delays by $\tau$, if the delay of defection is bigger than that of cooperation, then in the limit of infinite $\tau$, the interior unstable equilibrium tends to some value, both strategies have some basin of attraction, see Panel A. However, if the delay of cooperation is bigger than that of defection, then above some critical $\tau$, the unstable interior equilibrium ceases to exist and the defection becomes globally asymptotically stable, see Panel C. Under no circumstances cooperation becomes globally asymptotically stable.

\begin{figure}
   \centerline{\includegraphics[]{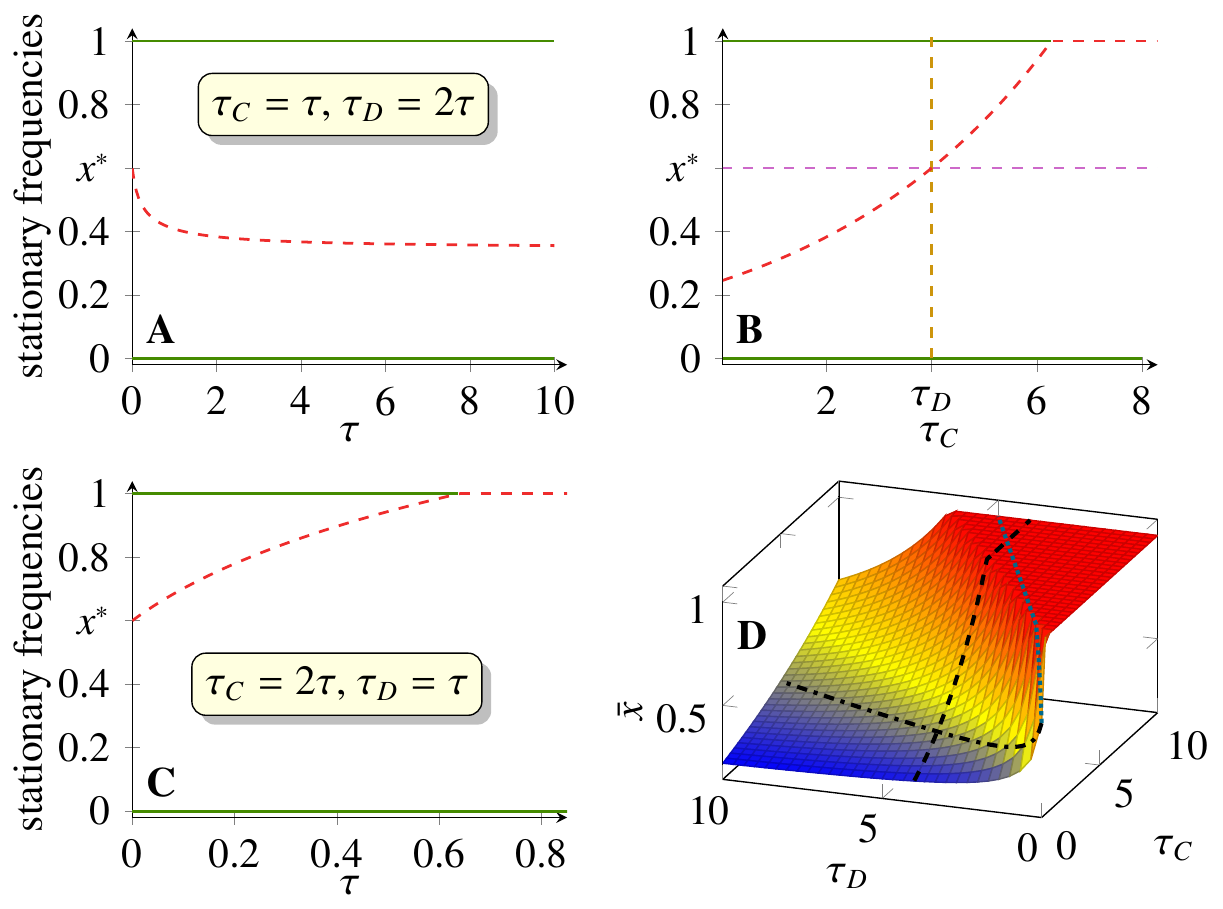}}
   \caption{Numerical solutions of $F(x)=0$ for the matrix $U_1$ for the classical Stug-hunt game, that is for $a=3$ and $\beta=3/5$. Panel \textbf{A}: the interior stationary state as a function of $\tau$, when $\tau_C=\tau$, $\tau_D=2\tau$. 
      Panel \textbf{B}: the interior stationary state as a function of $\tau_C$, while $\tau_D=4$ is fixed. Panel \textbf{c}: the interior stationary state as a function of $\tau$, 
      when $\tau_C=2\tau$, $\tau_D=\tau$. Panel \textbf{D}: the interior stationary state as a function of $\tau_C$ and $\tau_D$. The dash-dotted line indicates the cross-section presented on Panel \textbf{A}, the  black dotted line indicates the cross-section presented on Panel \textbf{B}, while the dotted line indicates the cross-section presented on Panel \textbf{C} .\label{fig.stag_hunt}}
\end{figure}

\subsection{Snowdrift-type games}
We consider here two examples of games with a unique stable interior equilibrium state. Because of nonlinearity of our equations we could not treat analytically general cases.
\vspace{2mm}

\noindent {\bf Example 1}
\vspace{2mm}

Here we discuss a classical Snowdrift game (with the reward equal to 5 and the cost equal to 2) with the following payoff matrix:
\[
U_2 = \begin{bmatrix}
4 & 3 \\ 5 & 0
\end{bmatrix}
\]

We have that $x^*=0.75$ is the stable stationary state of the non-delayed replicator dynamics. 

When delays are parametrized in the following form, $\tau_C=\tau, \tau_{D}=2\tau$ and $\tau$ increases, then the interior stationary state increases until it disappears 
at some value of $\tau$. Above that point, $x=1$, the equilibrium when all players cooperate, becomes globally asymptotically stable (see Fig.~\ref{fig.1}A). In the case of $\tau_C=2\tau, \tau_{D}=\tau$, 
the interior stationary state decreases and approaches the value $\frac{1+\sqrt{301}}{50}\approx0.37$ (see Fig.~\ref{fig.1}C).

We see in Fig.~\ref{fig.1}B, that for fixed $\tau_{D}=2$, there is a value of $\tau_{C}$ below which there is no interior stationary point; $x=1$ is globally asymptotically stable.
Above this point, the stable interior stationary state decreases. 

In Fig.~\ref{fig.1}D we present the unstable interior stationary state as function of $\tau_C$ and $\tau_D$. If $\tau_{C}=\tau_D$,
then the stable stationary frequency does not depend on time delay and is equal to $x^*=0.75$. For fixed $\tau_D$, 
if $\tau_{C}\gg \tau_D$, the stationary state decreases approaching the value $0.2$ (that can be easily calculated taking a limit $\tau_C\to+\infty$ in~\eqref{funss}). 
Similarly, analyzing the graph of the function $F$ for $\tau_D=0$, one can easily find that the stationary state decreases in this case from $0.75$ for $\tau_C=0$ to $0.2$ 
for $\tau_C\to+\infty$. On the other hand, if we fix $\tau_D$, then the stationary state is a decreasing function of $\tau_C$ (see Fig.~\ref{fig.1}B).  
Additionally, if $\tau_D$ is large enough (greater than $\ln(5/4)/3$, see Remark~\ref{rem:d0} in the Appendix), then there exists a value $\tau_C^*$ of $\tau_C$, 
such that for $\tau_C<\tau_C^*$ there exists no interior stable frequency and all individuals are playing C. 
The critical value $\tau_D^* = \frac{1}{3}\ln\frac{5}{4}$ for $\tau_C=0$ depends on $\tau_C$ almost linearly (see the implicit formula~\eqref{tauAstar} in the Appendix for the relation between~$\tau_{C}^*$ and $\tau_{D}^*$). 
\vspace{2mm}

We see that the situation here is in a sense a reverse one to that in the Stag-hunt game. When we parametrize both delays by $\tau$, if the delay of defection is bigger than that of cooperation,
then above some critical $\tau$, the stable interior equilibrium ceases to exist and the cooperation becomes a globally asymptotically stable equilibrium, see Panel A.
However, if the delay of cooperation is bigger than that of defection, then in the limit of infinite $\tau$, the stable interior equilibrium tends to some value, both strategies cooexist, see Panel C. 
Under no circumstances defection becomes globally asymptotically stable.

\begin{figure}
   \centerline{\includegraphics[]{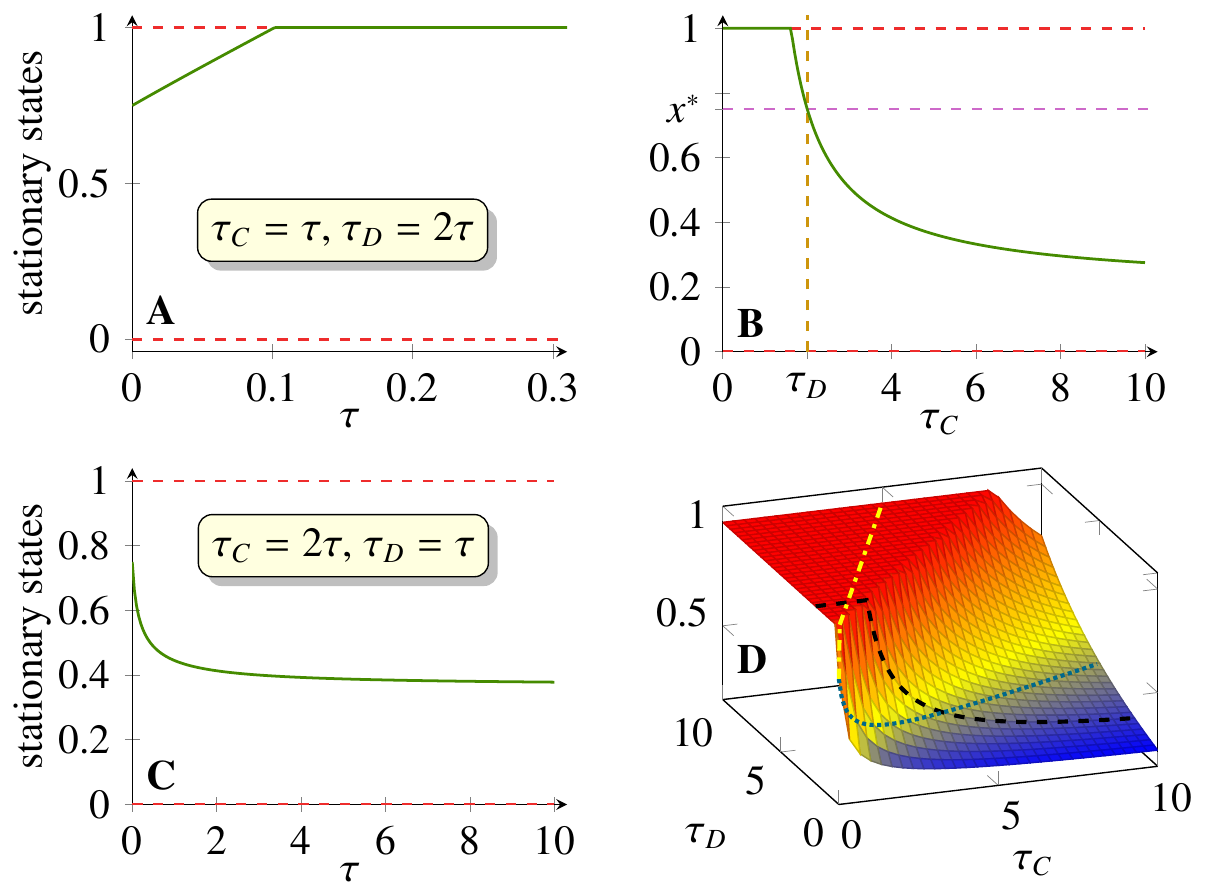}}
   \caption{Numerical solutions of $F(x)=0$ for the matrix $U_2$ for the classical Snowdrift game. Panel \textbf{A}: the unstable interior stationary state as a function of $\tau$, when $\tau_C=\tau$, $\tau_D=2\tau$. 
      Panel \textbf{B}: the interior stationary state as a function of $\tau_C$, while $\tau_D=2$ is fixed. Panel \textbf{C}: the interior stationary state as a function of $\tau$, 
when $\tau_C=2\tau$, $\tau_D=\tau$. Panel \textbf{D}: the interior stationary state as a function of $\tau_C$ and $\tau_D$.
The yellow dash-dotted line indicates the cross-section presented on Panel \textbf{A}, the black dotted line indicates the cross-section presented on Panel \textbf{B}, while the dotted line indicates the cross-section presented on Panel \textbf{C} }.\label{fig.1}
\end{figure}
\vspace{2mm}

\noindent \textbf{Example 2}
\vspace{2mm}

Here we study another particular example which illustrates a complex behaviour of games with asymmetric time delays.
The game has the following payoff matrix:
\[
U_3 = \begin{bmatrix}
2 & 0.5 \\ 2.95 & 0
\end{bmatrix}
\]
Now, $x^*\approx 0.345$ is the stable stationary state of the non-delayed replicator dynamics. 

In Fig~\ref{fig.2}A we set $\tau_C=\tau$, $\tau_D=2\tau$. We see that there exists a threshold $\tau^*\approx 1.1$ such that for $\tau<\tau^*$ 
there exists a unique interior stationary state. Numerical simulations suggest that this state is stable. For $\tau>\tau^*$  and $\tau< \tau^**$ there are two interior stationary states,
one of them stable, the other one unstable. Notice that $x=1$ is another stable equilibrium.
Numerical simulations suggest that if the initial frequency of the strategy C is large enough, then the D strategy is eliminated. 
Let us also observe that there two asymptotically stable equilibria in the system. 
Cooperation becomes an alternative equilibrium, then the only equilibrium when both delays increase.

\begin{figure}[htb]
   \centerline{\includegraphics[]{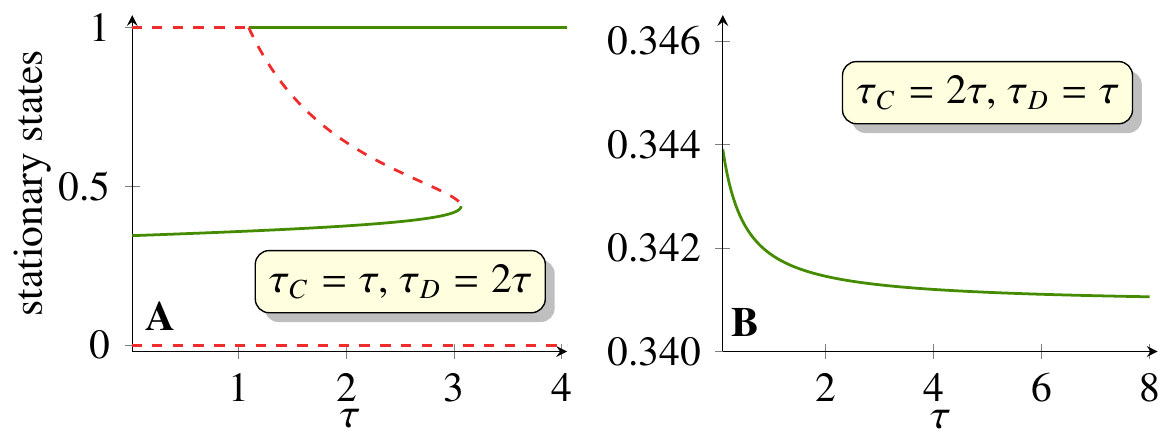}}
   \caption{Numerical solutions of $F(x)=0$ for the matrix $U_3$: the stable interior stationary state as a function of $\tau$, when $\tau_C=\tau$, $\tau_D=2\tau$ (Panel \textbf{A}), and 
      when $\tau_C=2\tau$, $\tau_D=\tau$ (Panel \textbf{B}). \label{fig.2}}
\end{figure}

In Fig.~\ref{fig.2}B we present the stable interior stationary state as a function of $\tau$,  $\tau_C=2\tau$, $\tau_D=\tau$. 
In this case, the stationary state is almost constant. In fact, for this set of parameters, one can easily see that the function $F$ is decreasing to one 
as both terms of $F$ decrease. Because only the first term depends on $\tau$ (and it decreases with increasing $\tau$), one can deduce
that the stationary state is a decreasing function of $\tau$. For $\tau\to+\infty$ it converges to the positive solution of the quadratic equation 
$(2.95 x)^2 - 1.5 x - 0.5 = 0$, so it is approximately equal to $0.341$. Thus, $\bar x\in [0.341,0.345]$, so it is almost constant. 
Here the states $1$ and $0$ are unstable.

\subsection{Prisoner's Dilemma game}

Here we consider the Prisoner's Dilemma game with the following payoff matrix:
\[
U_4 = \begin{bmatrix}
3 & 0 \\ 5 & 1
\end{bmatrix}.
\]

The function $F$ takes the following form,

\begin{equation}\label{fpd}
F(x) = \frac{1}{\tau_C-\tau_D} \ln\left( \frac{3x}{4x+1}\right) +1 - 3x \left(\frac{4x+1}{3x}\right)^{\frac{\tau_C}{\tau_C-\tau_D}}.
\end{equation}

It is easy to see that $\lim\limits_{x\to0^+}F(x)<0$ for $\tau_C>\tau_D$ and $\lim\limits_{x\to0^+}F(x)>0$ for $\tau_C<\tau_D$. However, the sign of 

\begin{equation}\label{fpd1}
F(1) = 1 - \frac{1}{\tau_C-\tau_D}\ln \frac{5}{3} - 3 \left(\frac{5}{3}\right)^{\frac{\tau_C}{\tau_C-\tau_D}}
\end{equation}
cannot be easily determined. Plots of the function $F$ suggest that it has no zeros inside the interval $[0,1]$ for $\tau_C>\tau_D$ and it can have one for some set of delays  $\tau_C<\tau_D$. 
However, we can deduce that $F(1)>0$ if $\tau_D$ is sufficiently large, thus then there exists an interior stationary frequency. In the Panel A of Fig.~\ref{fig.dylemat2} 
we plotted the value of this interior stationary frequency for $\tau_D=1$. The limiting values of $\tau_C^*$ and $\tau_D^*$ are given by the implicit formula $F(1) = 0$ which can be solved, 

\begin{equation}\label{solution}
\tau_D^* = \tau_C^* \left(1- \frac{\ln\frac{3}{5}}{\ln\left(\frac{1}{3\tau_C^*}W_L\left(3\tau_C^*\e^{\tau_C^*}\right)\right)}\right).
\end{equation}

The curve is plotted in Panel D of Fig.~\ref{fig.dylemat2} (the solid line). It is almost a straight line with the slope given by 

\begin{equation}\label{solution2}
1- \frac{\ln\frac{3}{5}}{\ln\left(\frac{1}{3\tau_C^*}W_L\left(3\tau_C^*\e^{\tau_C^*}\right)\right)} \to \frac{\ln 5}{\ln 3}\approx 1.465 \; \; as \tau_{C}^{*} \rightarrow \infty
\end{equation}

On the other hand, $\tau_D^*\to \frac12\log\frac53 \approx 0.255$ as $\tau_C^*\to 0$. If a point $(\tau_C,\tau_D)$ is above the curve $(\tau_C^*,\tau_D^*)$,
 then there exists an interior stationary state (which is unstable as suggested by numerical simulations). This means that the system is bistable and that it goes to all cooperation 
or all defection depending on the initial condition.

In the Panel C of Fig.~\ref{fig.dylemat2} we presented the interior stationary frequency for various $\tau_C's$ and $\tau_D's$. 

If $\tau_C=0$, then we have 

\begin{equation}\label{eqpd}
F(x) = \frac{1}{\tau_D} \ln\frac{4x+1}{3x} + 1 - 3x
\end{equation}
which is a decreasing function that has a zero inside the interval $[0,1]$ for $\tau_D>\frac{1}{2}\ln\frac{5}{3}$. A numerical solution of the equation 

\begin{equation}\label{numericpd}
\frac{1}{\tau_D} \ln\frac{4x+1}{3x} + 1 - 3x = 0
\end{equation}
is plotted in the Panel B of Fig.~\ref{fig.dylemat2}.

\begin{figure}
   \includegraphics[]{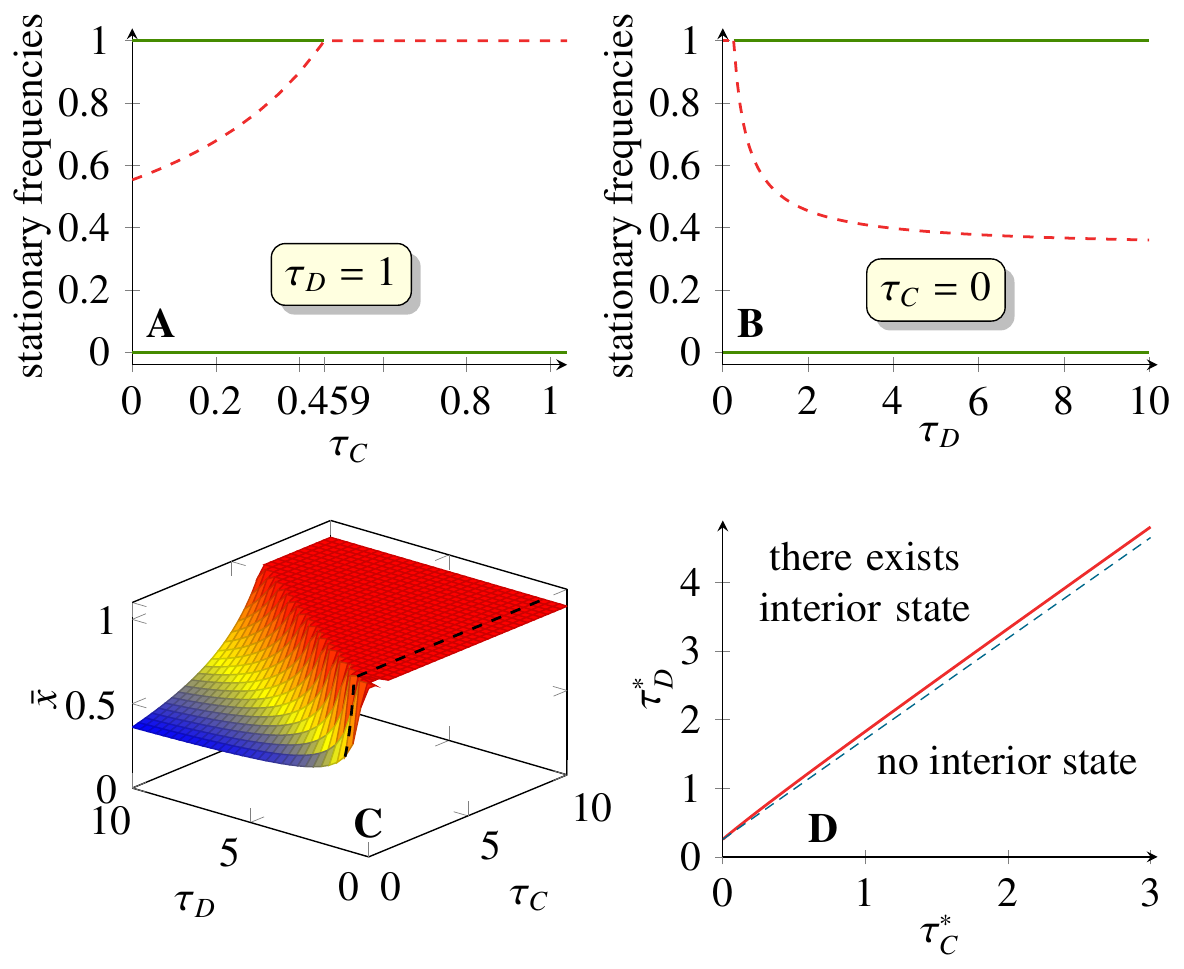}
   \caption{Numerical solutions of $F(x)=0$ for the matrix $U_4$ for the classical Prisoner's Dilemma game. Panel \textbf{A}: interior stationary state as a function of $\tau_C$, when $\tau_D=1$ is fixed. 
The interior stationary state exists for $\tau_C<0.459$.
   Panel \textbf{B}: the interior stationary state as a function of $\tau_D$, while $\tau_C=0$ is fixed. An interior stationary state exists for $\tau_D>0.2554$.
   Panel \textbf{C}: the interior stationary state as a function of $\tau_C$ and $\tau_D$.  The dotted line indicates the cross-section presented on Panel \textbf{A}.
   Panel \textbf{D}: the region of delays for which there exists the interior stationary state and the region where there is no interior stationary state. The solid curve denotes the border between these two regions. The curve is almost a straight line (the dashed line), plotted for a comparison, it has the same slope as the solid curve for $\tau_C^{*}\to+\infty$). 
      \label{fig.dylemat2}}
\end{figure}

We see that when one introduces asymmetric delays with a delay of cooperation smaller than a delay of defection, then there appears an interior unstable equilibrium. 
Hence the state with just cooperators becomes locally asymptotically stable.

\section{Discussion}

We studied effects of strategy-dependent time delays on stationary states of evolutionary games.

Recently, effects of the duration of interactions between two players on their payoffs, and therefore on evolutionary outcomes, were discussed by K\v{r}ivan and Cressman \cite{kricress}. 
In their models, the duration of interactions depend on the strategies involved. This can be interpreted as strategy-dependent time delays.
They showed that interaction times change stationary states of the system. 

Another approach is to consider ordinary differential equations with time delays. It was pointed out in \cite{alboszta} 
that in the so-called biological model, where it is assumed that the number of players born in a~given time is proportional to payoffs received by their parents at a~certain moment in the past, 
the interior state is asymptotically stable for any time delay. In such models, microscopic interactions lead to a new type of replicator dynamics which  
describe the time evolution of fractions of the population playing given strategies and the size of the population.

Here we studied biological-type models with strategy-dependent time delays. We observed a novel behavior. We showed that interior stationary states depend continuously on time delays. 

In particular, in games with two pure Nash equilibria (Stag-hunt-type games), the bigger the time delay of a given strategy is, the smaller its basin of attraction.

When we parametrize both delays by $\tau$, if the delay of defection is bigger than that of cooperation, then in the limit of infinite $\tau$, the interior unstable equilibrium tends to some value, both strategies have some basin of attraction. However, if the delay of cooperation is bigger than that of defection, then above some critical $\tau$, the unstable interior equilibrium ceases to exist and the defection becomes globally asymptotically stable. Under no circumstances cooperation becomes globally asymptotically stable.

In games with a stable interior equilibrium, the bigger the time delay of a strategy is, the less frequent is in the population.
Moreover, at certain time delays, the interior stationary state ceases to exist or there may appear another interior stationary state.

We see that the situation here is in a sense a reverse one to that in the Stag-hunt game. When we parametrize both delays by $\tau$, if the delay of defection is bigger than that of cooperation,
then above some critical $\tau$, the stable interior equilibrium ceases to exist and the cooperation becomes a globally asymptotically stable equilibrium.
However, if the delay of cooperation is bigger than that of defection, then in the limit of infinite $\tau$, the stable interior equilibrium tends to some value, both strategies cooexist. 
Under no circumstances defection becomes globally asymptotically stable.

In the Prisoner's Dilemma game, for time delays of cooperation smaller than the time delay of defection, there appears an unstable interior equilibrium 
and therefore for some initial conditions, the population converges to the homogeneous state with just cooperators. 
This shows an asymptotic stability of cooperation in a simple model of social dilemma.

It would be interesting to analyze strategy-dependent time delays in stochastic dynamics of finite populations. The work in this direction is in progress.

\appendix
\section*{Appendix}
\setcounter{equation}{0}
\setcounter{tw}{0}
\renewcommand{\theequation}{A.\arabic{equation}}
\renewcommand{\thetw}{A.\arabic{tw}}
\renewcommand{\thecor}{A.\arabic{cor}}
\renewcommand{\theprop}{A.\arabic{prop}}
\renewcommand{\thelem}{A.\arabic{lem}}
\renewcommand{\therem}{A.\arabic{rem}}

We provide here propositions, theorems, and their proofs which support results presented in the paper.   

The main results for Snowdrift-type games are included in Theorem~\ref{thm:brzeg} and following Corollary~\ref{cor:brzeg}. We present there a set of conditions that guarantee the existence of at least one interior stationary state. In Theorem~\ref{thm:monoton} we give a condition for the monotonicity of a function which zeros give us stationary states. Finally, in Proposition~\ref{prop:smalltauA} we state conditions for the absence of stationary states for small time delays. 

In Remark~\ref{rem:wl:brzeg} we give an explicit formula for the leading eigenvalue $\lambda$ for pure stationary states, $\bar x=0$ and $\bar x=1$. This eigenvalue determines the growth rate of the whole population. In Remark~\ref{rem:d0} we state some results about the existence of stationary states for $d=0$. 

We provide here also additional calculations concerning Stag-hunt type games.

\subsection{General results for Snowdrift-type games}\label{app1}

\medskip{}
Interior stationary states are given by zeros of the function $F(x)$ defined by~\eqref{funss}


Let $\bar U_C = a \bar x + b(1-\bar x)$ and $\quad \bar U_D = c \bar x + d (1-\bar x).$

\begin{rem}\label{rem:wl:brzeg}
   We see that for $\bar x= 0$ we have 
   \begin{equation}\label{lambdadla0}
      \lambda + 1 = d \e^{-\lambda\tau_D} \quad \Longrightarrow \quad \lambda = \frac{W_p(d \tau_D \e^{\tau_D})}{\tau_D},
   \end{equation}  
   while for $\bar x= 1$ we have 
   \begin{equation}\label{lambdadla1}
      \lambda + 1 = a \e^{-\lambda\tau_C} \quad \Longrightarrow \quad \lambda = \frac{W_p(a \tau_C \e^{\tau_C})}{\tau_C},
   \end{equation}  
   where $W_p$ is the Lambert $W$ function, that is a principle branch of the relation $W_p(x)\exp\bigl(W_p(x)\bigr)= x$. 
\end{rem}

We show that if $\tau_{C}\neq\tau_D$, then the value of an interior stationary state depends on time delays. 
Moreover, for some payoff matrices and some values of a delay, multiple interior stationary states exist. 
We would like to point out that these relations are not linear.
Thus, in contrary to the case without time delays or with equal delays (i.e. $\tau_C=\tau_{D}$), 
adding a constant to a column of a payoff matrix or multiplying the matrix by a constant, changes interior stationary states 
or even may change the number of interior stationary states. Thus, the intuition from non-delayed games about the asymptotic of dynamics 
of replicator equations are not necessarily valid for the case of non-equal delays.   

It turns out, that if the eigenvalue $\lambda(\bar x)$ corresponding to the stationary state $\bar x$ is positive, 
then the frequency of the given strategy is smaller than the frequency of this strategy in a non-delayed case if the delay corresponding to this strategy is larger than the time delay of the other strategy. 
We have the following proposition.
\begin{prop}\label{prop:eig}
   Let $a<c$ and $d<b$. Assume that $\bar x\in (0,1)$ and let $\lambda$ be the leading eigenvalue that corresponds to $\bar x$. Then if $\lambda>0$ we have 
   \begin{itemize}
      \item if $\tau_C>\tau_D$, then $\bar x<x^*$
      \item if $\tau_C<\tau_D$, then $\bar x>x^*$
   \end{itemize}
\end{prop}
\begin{proof}
   Assume that $\tau_C>\tau_D$. Then the sign of $\lambda$ is the same as the sign of 
   \[
   \phi(x)=\ln\left(\frac{a x + b(1-x)}{c x + d(1-x)}\right). 
   \]
   Note that $\phi(x^*)=0$ and 
   \[
   \phi'(x) = \frac{ad-bc}{\bigl(ax+b(1-x)\bigr)\bigl(cx+d(1-x)\bigr)} <0 
   \]
   because $ad<bc$. Thus, $\phi(x)>0$ for $x<x^*$, so $\bar x<x^*$. Similarly, if $\tau_C<\tau_D$, then the sign of $\lambda$ is opposite 
   to the sign of $\phi(x)$ and therefore $\bar x>x^*$.
\end{proof}

We will study general properties of $F$, they will help us to determine a number of stationary states of our replicator dynamics.
First, we determine conditions that would imply the sign of $F$ at $x=0$ and $x=1$.
Let us define 
\begin{equation}\label{adgw}
   c^*(a) = a^{-\tau_D/\tau_C} \left(\frac{W_p\bigl(a\tau_C\e^{\tau_C}\bigr)}{\tau_C}\right)^{1-\tau_D/\tau_C}, 
   \qquad 
   d^*(b) = b^{-\tau_D/\tau_C} \left(\frac{W_p\bigl(b\tau_C\e^{\tau_C}\bigr)}{\tau_C}\right)^{1-\tau_D/\tau_C},
\end{equation}
where $W_p$ is the Lambert $W$ function, that is a principle branch of the relation $W_p(x)\exp\bigl(W_p(x)\bigr)= x$. 
\begin{tw}\label{thm:brzeg}
   \textbf{\upshape (A)} If $a<c$, then $F(1)>0$ if and only if one of the following condition holds,
   \begin{enumerate}[\upshape (i)]
      \item $\tau_C>\tau_D$, $a<1$ and $c<c^*(a)$
      \item $\tau_C<\tau_D$, $a<1$ or $c>c^*(a)$
   \end{enumerate}
   \textbf{\upshape (B)}  If $b>d$, then $F(0)>0$ if and only if one of the following condition holds,
   \begin{enumerate}[\upshape (i)]
      \item $\tau_C>\tau_D$, $b<1$ or $d<b^*(b)$
      \item $\tau_C<\tau_D$, $b<1$ and $d>b^*(b)$
   \end{enumerate}
\end{tw}
\begin{proof}   
   
   First we study the sign of $F(1)$.  It is easy to see that 
   \[
   F(1) = \frac{1}{\tau_C-\tau_D}\ln\left(\frac{a}{c}\right) +1 -    a\biggl(\frac{c}{a}\biggr)^{\tau_C/(\tau_C-\tau_D)}.
   \]
   Assume that $\tau_C>\tau_D$. Due to the assumption $a<c$, there exists $z\in (0,1)$ such that 
   $a = z c$. We plug this into the expression for $F(1)$ and we obtain
   \begin{equation}\label{formulazz}
      \frac{1}{\tau_C-\tau_D}\ln z +1 -  a\biggl(\frac{1}{z}\biggr)^{\tau_C/(\tau_C-\tau_D)}.
   \end{equation}
   Let us introduce $\displaystyle u = z^{\tau_C/(\tau_C-\tau_D)}$. Then the above expression simplifies to
   \[
   f_a(u) = \frac{1}{\tau_C}\ln u +1 -\frac{a}{u}.
   \]
   It is easy to see that $f_a$ is a continuous, strictly increasing function of $u$ and $\displaystyle \lim_{u\to0^+}f_a(u) = -\infty$. Thus we have
   two possibilities. Either $f_a(1)< 0$ and $f_a(u)<0$ for all $u\in (0,1)$ (thus $F(1)<0$ for all $a<c$) or there exits a unique 
   $u^*\in (0,1)$ such that $f_a(u^*)=0$ ($f_a(u)<0$ for $0<u<u^*$ and $f_a(u)>0$ for $u^*<u<1$) due to monotonicity (and this implies an appropriate sign of $F(1)$ depending on the value 
   of $a$ with respect to~$c$). 
   
   Note that $f_a(1)\le 0$ is equivalent to $a\ge 1$. Assume now that $a<1$ and we find $u^*$. We have 
   \[
   \frac{1}{\tau_C}\ln u^* +1 = \frac{a}{u^*} \; \;\Longleftrightarrow\;\;
   \tau_C-a\tau_C\frac{1}{u^*} = \ln \frac{1}{u^*}  \; \;\Longleftrightarrow\;\;
   a\tau_C \frac{1}{u^*} \e^{a\tau_C\frac{1}{u^*}} = a\tau_C e^{\tau_C}. 
   \]
   Thus,
   \[
   u^* = \frac{a\tau_C}{W_p\bigl(a\tau_C\e^{\tau_C}\bigr)},
   \]
   where $W_p$ is the Lambert function. Because $\frac{a}{c} = u^{(\tau_C-\tau_D)/\tau_C}$, then $F(1)\le 0$ if $a\ge 1$ or $c> c^*(a)$ 
   where $c^*(a)$ is given by~\eqref{adgw}.  Thus, if $a>1$ and $c< c^*(a)$, then $F(1)>0$ and the point (i) is proved. Finally, note that 
   if $\tau_C<\tau_D$, then the variable $u$ changes from $1$ to $+\infty$ instead of from $0$ to $1$ and very similar arguments lead to 
   the assertion of the point (ii). 
   
   Let us calculate the value
   \[
   F(0) = \frac{1}{\tau_C-\tau_D}\ln\left(\frac{d}{b}\right) +1 -    b\biggl(\frac{b}{d}\biggr)^{\tau_C/(\tau_C-\tau_D)}.
   \]
   Note that this situation is analogous to the previous one and analogues arguments prove this part of the theorem.
\end{proof}

From Theorem~\ref{thm:brzeg} and its proof the following corollary can be easily deduced. 
\begin{cor}\label{cor:brzeg}
   \textbf{\upshape (A)} If $a<c$, then $F(1)<0$ if and only if one of the following condition holds:
   \begin{enumerate}[\upshape (i)]
      \item $\tau_C>\tau_D$, $a\ge 1$ or $c>c^*(a)$
      \item $\tau_C<\tau_D$, $a\ge 1$ and $c<c^*(a)$
   \end{enumerate}
   \textbf{\upshape  (B)}  If $b>d$, then $F(0)>0$ if and only if one of the following condition holds,
   \begin{enumerate}[\upshape (i)]
      \item $\tau_C>\tau_D$, $b\ge 1$ and $d>b^*(b)$
      \item $\tau_C<\tau_D$, $b\ge 1$ or $d<b^*(b)$
   \end{enumerate}
\end{cor}

\begin{cor}
   Theorem~\ref{thm:brzeg} and Corollary~\ref{cor:brzeg} give conditions that guarantee the existence of at least one interior stationary state. Namely, if $a<c$ and $b>d$, then we need $F(1)>0$ (Theorem~\ref{thm:brzeg}\textbf{(A)}) and $F(0)<0$ (Corollary~\ref{cor:brzeg}\textbf{(B)}) or, reversely, $F(1)<0$ (Corollary~\ref{cor:brzeg}\textbf{(A)}) and $F(0)>0$ (Theorem~\ref{thm:brzeg}\textbf{(B)}). 
\end{cor}

If the function $F$ is monotonic, then the previous theorem gives us a condition guaranteeing the existence of a solution of $F(x)=0$ on $(0,1)$, 
that is the existence of an interior stationary state.

\begin{tw}\label{thm:monoton}
   If $a<c$, $d<b$, and additionally $a<d<c$ or $d<a<b$, then the function $F$ is monotonic. Moreover, if additionally $\tau_C>\tau_D$, 
   then it is decreasing and if $\tau_C<\tau_D$, then it is increasing. 
\end{tw}
\begin{proof}
   It is enough to calculate the first derivative of $F$. It reads
   \[
   \begin{split}
      F'(x) &= \frac{1}{\tau_C-\tau_D}\Biggl( \frac{ad-bc}{\bigl(ax+b(1-x)\bigr)\bigl(cx+d(1-x)\bigr)}   +\\
      &\qquad +\biggl(\tau_C(d-c)+\tau_D(a-b)
      \frac{cx+d(1-x)}{ax+b(1-x)}\biggr)\biggl(\frac{cx+d(1-x)}{ax+b(1-x)}\biggr)^{\tau_D/(\tau_C-\tau_D)}\Biggr)
   \end{split}
   \]
   The assumptions guarantee that $ad-bc<0$, $d-c<0$, $a-b<0$, and the thesis follows.
\end{proof}

\begin{cor}
   If $a=b$ or $c=d$, then the function $F$ is monotonic and there is at most one solution of $F(x)=0$ in the interval $[0,1]$.
\end{cor}

\begin{rem}
   Theorems~\ref{thm:brzeg} and~\ref{thm:monoton} give a complete description of the existence of the unique stationary 
   state $\bar x$ inside the interval $(0,1)$ if $a<d<c$ or $d<a<b$ (under the assumption that $a<c$, $d<b$). 
\end{rem}

Now, we give a condition that guarantees the existence of the interior stationary point $\bar x$ when $\tau_{D}$ is fixed and $\tau_{C}$ converges to zero.
\begin{prop} \label{prop:smalltauA}
   Assume that $0<b<a<c$ and $d=0$. 
   If one of the following conditions holds,
   \begin{enumerate}[\bf (a)]
      \item $\displaystyle \tau_{D}\le \frac{b}{a(a-b)}$, $a>1$, and $\displaystyle \tau_{D}> \frac{\ln\left(\frac{c}{a}\right)}{a-1},$
      \item $\displaystyle \tau_{D}>\frac{b}{a(a-b)}$ and $b>1,$
   \end{enumerate}
   then for $\tau_C<\tau_{D}$ close enough to $0$, there exists no interior equilibrium state $\bar x\in (0,1)$. 
\end{prop}  
\begin{proof}
   We check if the equation $F(x)=0$, where $F$ is given by~\eqref{funss}, has a solution $\bar x\in (0,1)$ for fixed $\tau_{D}$ and a small $\tau_C$. 
   It is easy to see that the function $F$ is continuous with respect to $\tau_C$ and $x$ for $\tau_{C}<\tau_D$. Thus, it is enough to check the existence of 
   a solution to $F(x)=0$ for $\tau_C=0$. In this case, the function $F$ simplifies to
   \[
   F_0(x) = \frac{1}{\tau_D}\ln\left(\frac{c x}{a x + b(1-x)}\right) +1 - \bigl(a x + b(1-x)\bigr).
   \]
   We calculate the derivative of $F_0$, 
   \[
   F_0'(x) = \frac{b}{\tau_{D}} \cdot \frac{1}{x\bigl((a-b)x+b\bigr)} - (a-b).
   \]
   We see that it is a decreasing function of $x\in (0,1)$ (as $a>b$) and therefore $F_0$ is concave. Moreover, 
   $\lim\limits_{x\to 0+} F_0'(x) = +\infty$ thus 
   \begin{enumerate}
      \item if $\displaystyle\tau_{D}\le \frac{b}{a(a-b)}$, then $F_0$ is increasing in $(0,1)$;
      \item if $\displaystyle\tau_{D} > \frac{b}{a(a-b)}$, then $F_0$ has exactly one maximum at $(0,1)$ at the point 
      \[
      x_{\max{}} = \frac{b}{2(a-b)}\left(\sqrt{1+\frac{4}{b\tau_D}}-1\right).
      \]
   \end{enumerate}
   Let us consider the first case. Here $F_0'(x)>0$ for all $x\in (0,1)$ because it is decreasing and $F_0'(1)>0$. Thus $F_0$ is an increasing function of $x$ and as $\lim\limits_{x\to0+} F_0(x)=-\infty$ it has a zero in the interval $(0,1)$ if and only if $F(1)>1$. An easy calculation allows to derive the condition \textbf{(a)}.
   
   Now let us consider the second case. Some algebraic manipulations lead to 
   \[
   F_0(x_{\max{}}) = \frac{1}{\tau_D} \left(\ln \frac{c}{a-b} + \ln\biggl(\frac{4}{b\tau_D}\biggr) - 2\ln\left(1+\sqrt{1+\frac{4}{b\tau_D}}\right)\right)
   +1-\frac{b}{2}\left(1+\sqrt{1+\frac{4}{b\tau_D}}\right).
   \]
   We show that $F_0(x_{\max})$ approaches its supremum either for $\tau_D\to+\infty$ or for $\tau_D \to \frac{b}{a(a-b)}$. Let us introduce
   a new variable $z = \frac{4}{b\tau_D}$ and denote $\tilde c = \frac{c}{a-b}$. Now, writing $F_0(x_{\max{}})$ in variable $z$ we have 
   \[
   F_0(x_{\max{}}) = h(z) = \frac{b}{4}z \Bigl( \ln c + \ln z - 2\ln\bigl(1+\sqrt{1+z}\bigr)\Bigr) +1 - \frac{b}{2}\Bigl(1+\sqrt{1+z}\Bigr),
   \]
   where $\displaystyle 0<z< \frac{4(a-b)b}{b^2}$ as $\displaystyle \tau_{D} > \frac{b}{(a-b)b}$.
   We calculate the derivative of $h$ with respect to $z$ and obtain
   \[
   h'(z) = \frac{b}{4}\ln \left(\frac{c z}{(1+\sqrt{1+z})^2}\right).
   \]
   Now it is easy to see that $h'$ has at most one zero for $z>0$ and it is negative for $z$ close to $0$. Hence, the function $h$ is decreasing for 
   small $z$ and it may increase for large $z$ having at most one minimum for $z>0$. Thus it approaches its supremum either for $z\to 0$ (i.e. $\tau_D\to +\infty$) 
   or for $z\to  \frac{4(a-b)b}{b^2}$. In the latter case we have 
   \[
   \tau_{D}\to \frac{b}{a(a-b)}\; \Longrightarrow \; x_{\max{}}  \to 1,
   \]
   and we arrive at the case considered earlier. On the other hand,
   \[
   \lim_{z\to 0+} h(z) = 1-b < 0
   \]
   if the condition \textbf{(b)} holds. Thus, $F_0(x)<0$ for all $x\in(0,1)$ and no interior stationary state exists.
\end{proof}
\begin{rem}\label{rem:d0}
   Note that if  $\displaystyle \tau_{D}\le \frac{b}{a(a-b)}$ and either $a<1$ or  $\displaystyle \tau_{D}< \frac{\ln\left(\frac{c}{a}\right)}{a-1}$, then
   there exists exactly one interior stationary state $\bar x\in (0,1)$ for $\tau_C<\tau_D$ close enough to~$0$. On the other hand, if 
   $\displaystyle \tau_{D}> \frac{b}{a(a-b)}$ and $b<1$, then there exist one or two interior stationary states $\bar x\in (0,1)$ for $\tau_C<\tau_D$ close enough 
   to~$0$ if $\tau_D$ is large enough. In fact, if $a<1$ for a sufficiently large $\tau_D$, then there exist two stationary states $\bar x\in (0,1)$. 
\end{rem}

If $\tau_{C}$ is small enough (for a fixed $\tau_{D}$) or, reversely, if $\tau_{D}$ is large enough (for a fixed $\tau_C$), there exists no interior stationary state. Thus, it is possible to calculate these threshold values $\tau_{C}^*$ and $\tau_{D}^*$. It can be seen that the interior stationary state disappears when it merges with the stationary state $1$. Thus, looking for $\tau_{C}^*$ and $\tau_{D}^*$ such that $F(1)=0$, after some algebraic calculations we obtain the implicit formula
\begin{equation}\label{tauAstar}
   \tau_{C}^* = -\frac{(\tau_D^*-\tau_{C}^*)\ln\left(\frac{\ln(c/a)}{4(\tau_D^*-\tau_{C}^*)}+\frac{1}{a}\right)}{\ln(c/a)}.   
\end{equation}
We have that  for $\tau_C<\tau_{C}^* $ or for $\tau_D>\tau_{D}^*$ there exists no interior stationary state. 
\vspace{2mm}

\subsection{Calculations concerning Stag-hunt type game}\label{app2}
\vspace{1mm}

In this part of Appendix we provide some details of calculations that allowed us to formulate statements about the slope of the curve $(\tau_{C}^*,\tau_{D}^*)$ given by formula (18) in the paper. Let us remind this formula here, 
\begin{equation}\label{StagHuntGranica}
   \tau_{D}^* = \tau_{C}^*\left(1+\frac{\ln \beta}{W_p(a\tau_{C}^*\e^{\tau_C^*})-\tau_C^*}\right).
\end{equation}
First note, that $\tau_D^*>0$ for $\tau_C^* > \tau_{A,0}^*$, where 
\begin{equation}\label{tauA0def}
   \tau_{A,0}^* = \frac{-\ln \beta}{a\beta -1}.
\end{equation}
We also have $\tau_D^*=0$ for $\tau_C^* = \tau_{A,0}^*$. We have then
\begin{equation}\label{tauA0}
   W_p(a\tau_{A,0}^*\e^{\tau_{A,0}^*}) = \tau_{A,0}^*- \ln\beta.
\end{equation}

We calculate two quantities: 
\begin{enumerate}
   \item the derivative of the right-hand side of~\eqref{StagHuntGranica} with respect to $\tau_C^*$ at the point $\tau_{A,0}^*$;
   \item the slope of the right-hand side of~\eqref{StagHuntGranica} for $\tau_C^*\to+\infty$. 
\end{enumerate}

Here we calculate $g'(\tau_{A,0}^*)$, where 
\[
g(\tau_{C}^*) = \tau_{C}^*\left(1+\frac{\ln \beta}{W_p(y(\tau_{C}^*))-\tau_C^*}\right), \quad y(\tau_{C}^*) = a\tau_{C}^*\e^{\tau_C^*}.
\]
We have 
\begin{equation}\label{gprim}
   g'(\tau_{A,0}^*) = \left(1+\frac{\ln \beta}{W_p(y(\tau_{A,0}^*))-\tau_{A,0}^*}\right) - \frac{\tau_{A,0}^*\ln \beta}{\bigl(W_p(y(\tau_{A,0}^*))-\tau_{A,0}^*\bigr)^2}
   \Bigl(W_p'(y(\tau_{A,0}^*))y'(\tau_{A,0}^*)-1\Bigr).
\end{equation}
Because of~\eqref{tauA0}, the first parenthesis of~\eqref{gprim} is equal to 0. Moreover, $W_p(y(\tau_{A,0}^*))-\tau_{A,0}^*=-\ln \beta$. Thus, the formula~\eqref{gprim}
simplifies to
\begin{equation}\label{gprim1}
   g'(\tau_{A,0}^*) = - \frac{\tau_{A,0}^*}{\ln \beta}
   \Bigl(W_p'(y(\tau_{A,0}^*))y'(\tau_{A,0}^*)-1\Bigr).
\end{equation}
Now we calculate the expression in the parenthesis. In order to shorten the notation we write $y$ and $y'$ omitting dependence of $y$ on $\tau_{A,0}$. 
Using~\eqref{tauA0} and~\eqref{tauA0def} we get 
\[
W_p'(y) = \frac{W_p(y)}{y(1+W_p(y))} = \frac{-\frac{ln\beta}{a\beta-1}-\ln\beta}{y \frac{a\beta-1-a\beta\ln \beta}{a\beta-1}}=\frac{-a\beta \ln \beta}{y\bigl(a\beta-1-a\beta\ln \beta\bigr)} = \frac{a\beta-1}{a\beta-1-a\beta\ln\beta}\beta^{a\beta}.
\]
because 
\[
y = a\tau_{A,0}\e^{\tau_{A,0}^*} = \frac{-a\ln \beta}{a\beta-1} \beta^{1-a\beta}.
\]
On the other hand we have 
\[
y'(\tau_{A,0}^*) = \bigl(1+\tau_{A,0}^*\bigr)\e^{\tau_{A,0}^*} = \frac{\beta-1-\ln\beta}{a\beta-1}\beta^{1-a\beta}.
\]
Plugging formulas for $W_p'(y)$ and $y'$ into~\eqref{gprim1} we get the final formula 
\begin{equation}
   g'(\tau_{A,0}^*) = \frac{a\beta-1}{a\beta-1-a\beta\ln\beta}.
\end{equation}

In order to calculate the slope of the right-hand side of~\eqref{StagHuntGranica} for $\tau_C^*\to+\infty$ we note (see~\cite[Eq.~4.13.10]{DLMF}) that 
\[
W_p(y) = \ln y - \ln\ln y + r(y), \quad \text{ where } \quad r(y)\to 0 \text{ as } y\to+\infty. 
\]
We have then
\[
\begin{split}
   W_p\left(a\tau_C^*\e^{\tau_C^*}\right) - \tau_{C}^* &= \ln \left(a\tau_C^*\e^{\tau_C^*}\right) - \ln\ln\left(a\tau_C^*\e^{\tau_C^*}\right) + r\left(a\tau_C^*\e^{\tau_C^*}\right) - \tau_{C}^* \\      
   & = \ln a + \ln \frac{\tau_{C}^*}{\tau_C^*+\ln a+\ln\tau_{C}^*} + r\left(a\tau_C^*\e^{\tau_C^*}\right)\\
   & \to \ln a
\end{split}
\]
as $\tau_C^*\to +\infty$. Thus, the slope of the right-hand side of~\eqref{StagHuntGranica} for $\tau_C^*\to+\infty$ is equal to $1+\frac{\ln \beta}{\ln a}$. 

\section*{Acknowledgments}
We thank the National Science Centre, Poland, for a~financial support under the grant no.~2015/17/B/ST1/00693.  
We would like to thank two anonymous referees for valuable comments and suggestions which greatly improved our paper.


%
%
%


\begin{thebibliography}{99}

\bibitem{maynard1} J. Maynard Smith and G. R Price, The logic of animal conflict, Nature (London) 246, 15 (1973).
\bibitem{maynard2} J. Maynard Smith, {\em Evolution and the Theory of Games}, Cambridge University Press, Cambridge (1982).
\bibitem{weibull} J. Weibull, {\em Evolutionary Game Theory}, MIT Press, Cambridge MA (1995).
\bibitem{redondo} F. Vega-Redondo, {\em Evolution, Games, and Economic Behaviour}, Oxford University Press, Oxford (1996).
\bibitem{hofbauerbook} J. Hofbauer and K. Sigmund, {\em Evolutionary Games and Population Dynamics},  Cambridge University Press (1998).
\bibitem{young} H. P. Young, {\em Individual Strategy and Social Structure: an Evolutionary Theory of Institutions}, Princeton University Press, Princeton (1998).
\bibitem{cressman} R. Cressman, {em Evolutionary Dynamics and Extensive Form Games}, MIT Press, Cambridge (2003).
\bibitem{nowak} M. A. Nowak, {\em Evolutionary Dynamics: Exploring the Equations of Life}, Harvard University Press, Cambridge (2006).
\bibitem{sandholm} W. H. Sandholm, {\em Population Games and Evolutionary Dynamics}, MIT Press, Cambridge (2009).

\bibitem{taylor} P. D. Taylor and L. B. Jonker, Evolutionarily stable strategy and game dynamics, Math Biosci 40, 145-156 (1978).
\bibitem{hofbauer} J. Hofbauer, P. Shuster, and K. Sigmund, A~note on evolutionarily stable strategies and game dynamics, J. Theor. Biol. 81, 609-612 (1979).

\bibitem{ladas} I. Gy\"{o}ri and G. Ladas, {\em Oscillation Theory of Delay Differential Equations with Applications}, Clarendon (1991).
\bibitem{gopalsamy} K. Gopalsamy, {\em Stability and Oscillations in Delay Differential Equations of Population}, Springer (1992).
\bibitem{kuang} Y. Kuang, {\em Delay Differential Equations with Applications in Population Dynamics}, Academic Press Inc. (1993).
\bibitem{erneux} T. Erneux, {\em Applied Delay Differential Equations}, Springer (2009).

\bibitem{tao} Y. Tao and Z. Wang, Effect of time delay and evolutionarily stable strategy, J. Theor. Biol. 187, 111-116 (1997).
\bibitem{alboszta} J. Alboszta and J. Mi\c{e}kisz, Stability of evolutionarily stable strategies in discrete replicator dynamics with time delay, J. Theor. Biol. 231, 175-179 (2004).
\bibitem{aoku} H. Oaku, Evolution with delay, Japan Econ Review 53, 114-133 (2002).
\bibitem{ijima1} R. Iijima, Heterogeneous information lags and evolutionary stability, Math Soc Sciences 63, 83-85 (2011).
\bibitem{ijima2} R. Iijima, On delayed discrete evolutionary dynamics, J. Theor. Biol. 300, 1-6 (2012).

\bibitem{moreira} J. A. Moreira, F. L. Pinheiro, N. Nunes, and J. M. Pacheco, Evolutionary dynamics of collective action when individual fitness derives 
from group decisions taken in the past, J. Theor. Biol. 298, 8-15 (2012).
\bibitem{wesol} J. Mi\c{e}kisz and S. Weso\l owski, Stochasticity and time delays in evolutionary games, Dyn. Games Appl. 1, 440-448 (2011).
\bibitem{matusz} J. Mi\c{e}kisz, M. Matuszak, and J. Poleszczuk, Stochastic stability in three-player games with time delays, Dyn. Games Appl. 4, 489-498 (2014).
\bibitem{wessonrand1} E. Wesson and R. Rand, Hopf bifurcations in delayed rock-paper-scissors replicator dynamics, Dyn. Games Appl. 6, 139-156 (2016).
\bibitem{wessonrand2} E. Wesson, R. Rand, and D. Rand, Hopf bifurcations in two-strategy delayed replicator dynamics, J Bifurcation Chaos 26, 1650006 1-13 (2016).
\bibitem{nesrine2} N. Ben Khalifa, R. El-Azouzi, and Y. Hayel, Discrete and continuous distributed delays in replicator dynamics, Dyn. Games App. 8, 713-732 (2018).
\bibitem{nesrine1} N. Ben Khalifa, R. El-Azouzi, Y. Hayel, and I. Mabrouki, Evolutionary games in interacting communities, Dyn. Games Appl. 7, 131-156 (2016).

\bibitem{bmv} M. Bodnar, J. Mi\c{e}kisz, and R. Vardanyan, Three-player games with strategy-dependent time delays, Dyn. Games Appl. 10, 664–675 (2020).

\bibitem{bodnar1} M.~Bodnar, On the nonnegativity of solutions of delay differential equations,  Appl. Math. Lett., 13, 6,  91--95 (2000).

\bibitem{kricress} V. K\v{r}ivan and R. Cressman, Interaction times change evolutionary outcome: Two-player matrix games, J. Theor. Biol. 416, 199-207 (2017).

\bibitem{DLMF} \textit{NIST Digital Library of Mathematical Functions.} http://dlmf.nist.gov/, Release 1.0.28 of 2020-09-15. F. W. J. Olver, A. B. Olde Daalhuis, D. W. Lozier, B. I. Schneider, R. F. Boisvert, C. W. Clark, B. R. Miller, B. V. Saunders, H. S. Cohl, and M. A. McClain, eds.
\end{thebibliography}
\end{document}